\documentclass{IEEEtran}

 \IEEEoverridecommandlockouts

\usepackage{amsfonts,amssymb}
\usepackage{amsmath}
\usepackage{comment}

\usepackage{graphicx}

\newtheorem{theorem}{Theorem}
\newtheorem{lemma}[theorem]{Lemma}

\newtheorem{corollary}[theorem]{Corollary}

\newtheorem{assumption}{Assumption}

\newenvironment{proof}[1][\scshape{Proof:}]
{\begin{trivlist} \item[\hskip \labelsep {\small #1}]}
{\begin{flushright}$\blacksquare$\end{flushright}\end{trivlist}}

\def\lm{ \mu }

\def\TURN{ { \gamma } }
\def\p1{ { g_1 } }
\def\l1{ { l } }

\def\R{ {\mathbb{R} }}

\def\P{{P}}
\def\E{{E}}
\def\A{ {\mathbb{A} }}

\begin{document}

\title{Guesswork, large deviations and Shannon entropy}
\author{Mark M. Christiansen and Ken R. Duffy
\thanks{M. Christiansen and K. R. Duffy are with the Hamilton Institute,
National University of Ireland Maynooth.  Supported by the Irish
Higher Educational Authority (HEA) PRTLI Network Mathematics Grant.}}

\maketitle

\begin{abstract}
How hard is it guess a password? Massey showed that the Shannon
entropy of the distribution from which the password is selected is
a lower bound on the expected number of guesses, but one which is
not tight in general. In a series of subsequent papers under ever
less restrictive stochastic assumptions, an asymptotic relationship
as password length grows between scaled moments of the guesswork
and specific R\'{e}nyi entropy was identified.

Here we show that, when appropriately scaled, as the password length
grows the logarithm of the guesswork satisfies a Large Deviation
Principle (LDP), providing direct estimates of the guesswork
distribution when passwords are long. The rate function governing
the LDP possess a specific, restrictive form that encapsulates
underlying structure in the nature of guesswork. Returning to
Massey's original observation, a corollary to the LDP shows that
expectation of the logarithm of the guesswork is the specific Shannon
entropy of the password selection process.

\end{abstract}


\section{Introduction}
If a password, $W$, is chosen at random from a finite set
$\A=\{1,\ldots,m\}$, how hard is it to guess $W$? If $\{P(W=w)\}$
is known, then an optimal strategy is to guess passwords in decreasing
order of probability. Let $G(w)$ denote the number of attempts
required before correctly guessing $w\in\A$, called $w$'s guesswork.
Massey \cite{Massey94} proved that the Shannon entropy of $W$ is a
lower bound on the expected guesswork, $\E(G(W))$, and that no
general upper bound exists. This raised serious questions about the
appropriateness of Shannon entropy as a measure of complexity of a
distribution with regards guesswork. As a corollary to stronger
results, in this article we identify a large password relationship
between the expectation of the logarithm of the guesswork and
specific Shannon entropy.

Arikan \cite{Arikan96} introduced an asymptotic regime for studying
this problem by considering a sequence of passwords, $\{W_k\}$,
with $W_k$ chosen from $\A^k$ with i.i.d. letters. Again guessing
potential passwords in decreasing order of probability for each $k$,
he related the asymptotic fractional moments of the guesswork to
the R\'enyi entropy of a single letter,
\begin{align*}
\lim_{k\to\infty} \frac 1k \log \E(G(W_k)^\alpha) = 
(1+\alpha) \log\sum_{w\in\A} P(W_1=w)^{\frac{1}{1+\alpha}}
\end{align*}
for $\alpha>0$, where the right hand side is $\alpha$ times the
R\'enyi entropy of $W_1$ evaluated at $1/(1+\alpha)$. This result
was subsequently extended by Malone and Sullivan \cite{Malone04}
to word sequences with letters chosen by a Markov process and,
further still, by Pfister and Sullivan \cite{Pfister04} to sophic
shifts whose shift space satisfies an entropy condition and whose
marginals possess a limit property. Recently, using a distinct
approach Hanawal and Sundaresan \cite{Hanawal11} provided alternate
sufficient conditions for the existence of the limit. In all cases,
the limit is identified in terms of the specific R\'enyi entropy
\begin{align}
\label{eq:guess}
\lim_{k\to\infty}\frac 1k \log \E(G(W_k)^\alpha) = 
\alpha\lim_{k\to\infty} \frac{1}{k} R_k\left(\frac{1}{1+\alpha}\right),
\end{align}
where $R_k(\alpha)$ is the R\'enyi entropy of $W_k$
\begin{align*}
R_k(\alpha) 
	= \frac{1}{1-\alpha}\log\left(\sum_{w\in\A^k} P(W_k=w)^\alpha\right).
\end{align*}

Here we shall assume the existence of the limit on the left hand
side of equation \eqref{eq:guess} for all $\alpha>-1$, its equality
with $\alpha$ times specific R\'enyi entropy, its differentiability
with respect to $\alpha$ in that range and a regularity condition
on the probability of the most-likely word, that $\lim k^{-1} \log
P(G(W_k)=1)$ exists. From this, Theorem \ref{thm:main} deduces that
the sequence $\{k^{-1}\log G(W_k)\}$ satisfies a Large Deviation
Principle (LDP) (e.g. \cite{Dembo}) with a rate function $\Lambda^*$
that must possess a specific form that will have a physical
interpretation: $\Lambda^*$ is continuous where finite, can be
linear on an interval $[0,a]$, for some $a\in[0,\log(m)]$, and then
must be strictly convex while finite on $[a,\log(m)]$.

In contrast to earlier results, Corollary \ref{cor:curious} to the
LDP gives direct estimates on the guesswork distribution $P(G(W_k)=n)$
for large $k$, suggesting the approximation
\begin{align}
\label{eq:wknapprox}
P(G(W_k)=n) \approx \frac 1n \exp(-k\Lambda^*(k^{-1}\log n)). 
\end{align}
As this calculation only involves the determination of $\Lambda^*$,
to approximately calculate the probability of the $n^{\rm th}$ most
likely word in words of length $k$ one does not have to identify
the word itself, which would be computationally cumbersome, particularly
for non-i.i.d. word sources.

Corollary \ref{cor:shannon} to the LDP recovers a r\^ole for Shannon
entropy in the asymptotic analysis of guesswork. It shows that the
scaled expectation of the logarithm of the guesswork converges to
specific Shannon entropy
\begin{align*}
\lim_{k\to\infty}\frac 1k \E(\log G(W_k)) = 
\lim_{k\to\infty}\frac 1k H(W_k),
\end{align*}
where
\begin{align*}
H(W_k):= \sum_{w\in\A^k} \P(W_k=w)\log\P(W_k=w).
\end{align*}

\section{A Large Deviation Principle}

Consider the sequence of random variables $\{k^{-1}\log G(W_k)\}$.
Our starting point is the observation that the left hand side of
\eqref{eq:guess} is the scaled Cumulant Generating Function (sCGF)
of this sequence:
\begin{align*}
\Lambda(\alpha):=
	\lim_{k \rightarrow \infty}\frac{1}{k}\log 
	\E\left(e^{\alpha\log G(W_k)}\right),
\end{align*}
which is shown to exist for $\alpha>0$ in \cite{Arikan96}\cite{Malone04}
and for $\alpha>-1$ in \cite{Pfister04}.
\begin{assumption}
\label{ass:1}
For $\alpha>-1$, the sCGF $\Lambda(\alpha)$ exists, is equal to
$\alpha$ times the specific R\'enyi entropy, and has a continuous
derivative in that range.
\end{assumption}
We also assume the following regularity condition on the probability
of the most likely word.
\begin{assumption}
\label{ass:2}
The limit 
\begin{align}
\label{def:p1}
\p1=\lim_{k\to\infty} \frac 1k \log P(G(W_k)=1) 
\end{align}
exists in $(-\infty,0]$.
\end{assumption}
This assumption is transparently true for words constructed of
i.i.d. or Markovian letters. 

We first show that the sCGF exists everywhere.
\begin{lemma}[Existence of the sCGF]
\label{lem:flatsCGF}
Under assumptions \ref{ass:1} and \ref{ass:2}, for all $\alpha\leq-1$
\begin{align*}
\Lambda(\alpha)
= \lim_{k \rightarrow \infty}\frac{1}{k}\log \P(G(W_k)=1) 
= \p1
= \lim_{\beta\downarrow -1} \Lambda(\beta).
\end{align*}
\end{lemma}
\begin{proof}
Let $\alpha\leq-1$ and note that
\begin{align*}
&\log P(G(W_k)=1) 
\leq 
\log \sum_{i=1}^{m^k}P(G(W_k)=i)i^\alpha\\
&=
\log \E\left(e^{\alpha\log G(W_k)}\right)
\leq \log P(G(W_k)=1) + \log \sum_{i=1}^{\infty}i^\alpha.
\end{align*}
Taking $\liminf_{k\to\infty}k^{-1}$ with the first inequality and
$\limsup_{k\to\infty}k^{-1}$ with the second while using the Principle
of the Largest Term, \cite[Lemma 1.2.15]{Dembo} and usual estimates
on the harmonic series, we have that
\begin{align*}
\lim_{k\to\infty} \frac 1k \log \E(e^{\alpha\log G(W_k)})
	= \lim_{k\to\infty} \frac 1k \log P(G(W_k)=1) 
\end{align*}
for all $\alpha\leq-1$. 

As $\Lambda$ is the limit of a sequence
of convex functions and is finite everywhere, it is continuous
and therefore $\lim_{\beta\downarrow-1}\Lambda(\beta)=\Lambda(-1)$. 
\end{proof}
Thus the sCGF $\Lambda$ exists and is finite for all $\alpha$, with
a potential discontinuity in its derivative at $\alpha=-1$. This
discontinuity, when it exists, will have a bearing on the nature
of the rate function governing the LDP for $\{k^{-1}\log G(W_k)\}$.
Indeed, the following quantity will play a significant r\^ole in
our results:
\begin{align}
\label{eq:turn}
\TURN:=\lim_{\alpha \downarrow -1}\frac{d}{d\alpha}\Lambda(\alpha).
\end{align}
We will prove that the number of words with approximately equal
highest probability is close to $\exp(k\TURN)$. In the special case
where the $\{W_k\}$ are constructed of i.i.d. letters, this is
exactly true and the veracity of the following Lemma can be verified
directly.
\begin{lemma}[The number of most likely words]
\label{lem:iid}
If $\{W_k\}$ are constructed of i.i.d. letters, then
\begin{align*}
\TURN &= 
	\lim_{\alpha \downarrow -1} 
	\frac{d}{d\alpha} \alpha R_1((1+\alpha)^{-1})\\
	&= \log |\{w:P(W_1=w)=P(G(W_1)=1)\}|,
\end{align*}
where $|\cdot|$ indicates the number of elements in the set.
\end{lemma}
This i.i.d. result doesn't extend directly to the non-i.i.d. case
and in general Lemma \ref{lem:iid} can only be used to establish a
lower bound on $\TURN$:
\begin{align}
\label{eq:neqTURN}
\TURN=\lim_{\alpha \downarrow -1}\frac{d}{d\alpha}\Lambda(\alpha)
\ge \limsup_{k\to\infty} \lim_{\alpha\downarrow -1} 
\frac{d}{d\alpha} \alpha R_k((1+\alpha)^{-1}),
\end{align}
e.g \cite[Theorem 24.5]{Rockafellar70}. This lower bound
can be loose, as can be seen with the following example. Consider the
sequence of distributions for some $\epsilon>0$
\begin{align*}
P(W_k=i) = \begin{cases}
	m^{-k}(1+\epsilon) & \text{if } i=1	\\
	m^{-k}(1-\epsilon(m^k-1)^{-1})) & \text{otherwise}.
	\end{cases}
\end{align*}
For each fixed $k$ there is one most likely word and we have
$\log(1)=0$ on the right hand side of equation \eqref{eq:neqTURN}
by Lemma \ref{lem:iid}. The left hand side, however, gives $\log(m)$.
Regardless, this intuition guides our understanding of $\TURN$, but
the formal statement of it approximately capturing the number of
most likely words will transpire to be
\begin{align*}
\p1=
	\lim_{k\to\infty} 
	\frac 1k \log \inf_{\{w:G(w)<\exp(k\TURN)\}} P(W_k=w),
\end{align*}
where $\p1$ is defined in equation \eqref{def:p1}.

We define the candidate rate function as the Legendre-Fenchel
transform of the sCGF
\begin{align*}
\Lambda^*(x) &:= \sup_{\alpha\in\R} \{x\alpha -\Lambda(\alpha)\}\\
	&=
	\begin{cases}
	-x-\p1 & \text{ if } x\in[0,\TURN]\\
	\sup_{\alpha\in\R} \{x\alpha -\Lambda(\alpha)\} 
		& \text{ if } x\in(\TURN,\log(m)].
	\end{cases}
\end{align*}
The LDP cannot be proved directly by Baldi's version of the
G\"artner-Ellis theorem \cite{Baldi88}\cite[Theorem 4.5.20]{Dembo}
as $\Lambda^*$ does not have exposing hyper-planes for $x\in[0,\TURN]$.
Instead we use a combination of that theorem with the methodology
described in detail in \cite{Lewis95A} where, as our random variables
are bounded $0\leq k^{-1} \log G(W_k)\leq \log(m)$, in order to
prove the LDP it suffices to show that the following exist in
$[0,\infty]$ for all $x\in[0,\log m]$ and equals $-\Lambda^*(x)$:
\begin{align}
\label{eq:RL}
&\lim_{\epsilon \downarrow 0} 
	\liminf_{k \rightarrow \infty}\frac{1}{k}\log
	P\left(\frac 1k \log(G(W_k)) \in B_\epsilon(x)\right)\nonumber\\
&=\lim_{\epsilon \downarrow 0} 
	\limsup_{k \rightarrow \infty}\frac{1}{k}\log
	P\left(\frac 1k \log(G(W_k)) \in B_\epsilon(x)\right),
\end{align}
where $B_\epsilon(x)=(x-\epsilon,x+\epsilon)$. 

\begin{theorem}[The large deviations of guesswork] 
\label{thm:main}
Under assumptions \ref{ass:1} and \ref{ass:2}, the
sequence $\{k^{-1}\log G(W_k)\}$ satisfies a LDP with rate function
$\Lambda^*$.
\end{theorem}
\begin{IEEEproof}
To establish \eqref{eq:RL} we have separate arguments depending on
$x$. We divide $[0,\log(m)]$ into two parts: $[0,\TURN]$ and
$(\TURN,\log(m)]$. Baldi's upper bound holds for any $x\in[0,\log(m)]$.
Baldi's lower bound applies for any $x\in(\TURN,\log(m)]$ as 
$\Lambda^*$ is continuous and, as $\Lambda(\alpha)$  has a continuous
derivative for $\alpha>-1$, it only has a finite number of points
without exposing hyper-planes in that region.  For $x \in [0,\TURN]$,
however, we need an alternate lower bound.

Consider $x \in [0, \TURN]$ and define the sets
\begin{align*}
K_k(x,\epsilon) := 
	\left\{w\in\A^k:k^{-1}\log G(w) \in B_\epsilon(x)\right\},
\end{align*}
letting $|K_k(x,\epsilon)|$ denote the number of elements in each
set. We have the bound
\begin{align*}
&|K_k(x,\epsilon)|\inf_{w \in K_k(x,\epsilon)}\P(W_k=w)\\
&\le P\left(\frac 1k\log G(W_k) \in B_\epsilon(x)\right).
\end{align*}
As $\lfloor e^{k(x-\epsilon)} \rfloor \le
|K_k(x,\epsilon)| \le \lceil e^{k(x+\epsilon)} \rceil$,
we have that
\begin{align}
\label{eq:Kk}
x= \lim_{\epsilon\to0}\lim_{k\to\infty}\frac 1k \log |K_k(x,\epsilon)|.
\end{align}
By Baldi's upper bound, we have that
\begin{align*}
\lim_{\epsilon \downarrow 0} \limsup_{k \rightarrow \infty}
	\frac 1k \log P\left(\frac 1k\log G(W_k) \in B_\epsilon(x)\right)
&\leq x+\p1.
\end{align*}
Thus to complete the argument, for the complementary lower bound
we need to show that for any $x\in[0,\TURN]$
\begin{align*}
\lim_{\epsilon \downarrow 0} \liminf_{k \rightarrow \infty}
	\inf_{w \in K_k(x,\epsilon)}\frac{1}{k}\log\P(W_k=w) =\p1.
\end{align*}
If $\Lambda^*(x)<\infty$ for some $x>\TURN$, then for $\epsilon>0$
sufficiently small let $x^*$ be such that $\Lambda^*(x^*)<\infty$
and $x^*-\epsilon>\max(\TURN,x+\epsilon)$. Then by Baldi's lower bound,
which applies as $x^*\in(\TURN,\log(m)]$, we have
\begin{align*}
-\inf_{y\in B_\epsilon(x^*)}\Lambda^*(y)
\le 
	\liminf_{k\to\infty} \frac 1k \log 
	P\left(\frac 1k\log G(W_k) \in B_\epsilon(x^*)\right).
\end{align*}
Now
\begin{align*}
&P\left(\frac 1k\log G(W_k) \in B_\epsilon(x^*)\right)\\
&\le |K_k(x^*,\epsilon)|\sup_{w \in K_k(x^*,\epsilon)}\P(W_k=w)\\
&\le |K_k(x^*,\epsilon)|\inf_{w \in K_k(x,\epsilon)}\P(W_k=w),
\end{align*}
where in the last line we have used the monotonicity of guesswork
and the fact that $x^*-\epsilon>x+\epsilon$. Taking lower limits
and using equation \eqref{eq:Kk} with $|K_k(x^*,\epsilon)|$,
we have that
\begin{align*}
-\inf_{y\in B_\epsilon(x^*)}\Lambda^*(y)
&\le
	x^* + \liminf_{k \to\infty}
	\inf_{w \in K_k(x,\epsilon)}\frac{1}{k}\log\P(W_k=w)
\end{align*}
for all such $x^*,x$. Taking limits as $\epsilon\downarrow0$
and then limits as $x^*\downarrow\TURN$ we have
\begin{align*}
-\lim_{x^*\downarrow\TURN}\Lambda^*(x^*)\le 
	\TURN + 
\lim_{\epsilon \downarrow 0} \liminf_{k \rightarrow \infty}
	\inf_{w \in K_k(x,\epsilon)}\frac{1}{k}\log\P(W_k=w),
\end{align*}
but $\lim_{x^*\downarrow\TURN}\Lambda^*(x^*)=-\TURN-\p1$ so that
\begin{align*}
\lim_{\epsilon \downarrow 0} \liminf_{k \rightarrow \infty}
	\inf_{w \in K_k(x,\epsilon)}\frac{1}{k}\log\P(W_k=w)
	=\p1,
\end{align*}
as required.

Only one case remains: if $\Lambda^*(x)=\infty$ for all $x>\TURN$,
then we require an alternative argument to ensure that
\begin{align*}
\liminf_{k \to\infty}
	\inf_{w \in K_k(x,\epsilon)}\frac{1}{k}\log\P(W_k=w)
	= \p1.
\end{align*}
This situation happens if, in the limit, the distribution of words
is near uniform on the set of all words with positive probability.
Thus define
\begin{align*}
\lm :=\limsup_{k\to\infty}\frac 1k \log |\{w:P(W_k=w)>0\}|.
\end{align*}
As $\Lambda^*(x)=\infty$ for all $x>\TURN$, $\lm\leq \TURN$.  To
see $\TURN=\lm$, note that $\TURN= \lim_{\alpha \downarrow
-1}\Lambda'(\alpha)\leq \Lambda'(0)$. As
both $\Lambda(\alpha)$ and $\alpha R_k((1+\alpha)^{-1})$ are finite
and differentiable in a neighborhood of $0$, by \cite[Theorem
25.7]{Rockafellar70}
\begin{align*}
\Lambda'(0) 
= \lim_{k\to\infty} \frac 1k \frac{d}{d\alpha} 
	\alpha R_k((1+\alpha)^{-1})|_{\alpha=0} 
	= \lim_{k\to\infty} \frac 1k H(W_k).
\end{align*}
and $\lim_{k\to\infty} k^{-1} H(W_k) \leq \lm$. Thus $\TURN=\lm$ and,
due to convexity, $\Lambda$ is linear with slope $\lm$ on
$\alpha\in(-1,0]$. As $\Lambda(0)=0$, using Lemma \ref{lem:flatsCGF}
we have that $\p1=-\lm$. Let $x<\lm$ and consider
\begin{align*}
\l1 &=
\limsup_{k \to\infty}
        \sup_{w \in K_k(x+2\epsilon,\epsilon)}\frac{1}{k}\log\P(W_k=w)\\
        &\leq
\liminf_{k \to\infty}
        \inf_{w \in K_k(x,\epsilon)}\frac{1}{k}\log\P(W_k=w).
\end{align*}
We shall assume that $\l1<\p1$ and show this results in a contradiction.
Let $\epsilon<\min(\p1-\l1,\lm-x)/2$, then there
exists $N_\epsilon$ such that
\begin{align*}
\sum_{w\in\A^k} P(W_k=w) 
	&\leq e^{k(x+\epsilon)} e^{k(\p1+\epsilon)}
	+ e^{k(\lm+\epsilon)} e^{k(\l1+\epsilon)}\\
	&= e^{k(-\lm+x+2\epsilon)} + e^{k(-\p1+\l1+2\epsilon)},
\end{align*}
for all $k>N_\epsilon$, but this is strictly less than $1$ for $k$
sufficiently large and thus $\l1=\p1$. Finally, for $x=\lm$, 
and $\epsilon>0$,
note that we can decompose
$[0,\log(m)]$ into three parts, 
$[0,\lm-\epsilon]\cup(\lm-\epsilon,\lm+\epsilon)\cup[\lm+\epsilon,\log(m)]$,
where the scaled probability of the guesswork being in either the
first or last set is decaying, but
\begin{align*}
0 &= 
\lim_{k\to\infty} \frac 1k \log P\left(\frac 1k \log G(W_k)\in[0,\log(m)]\right)
\end{align*}
and so the result follows from an application of the principle of
the largest term.

Thus for any $x\in[0,\log(m)]$,
\begin{align*}
&\lim_{\epsilon \downarrow 0} 
	\liminf_{k \rightarrow \infty}\frac{1}{k}\log
	P\left(\frac 1k \log(G(W_k)) \in B_\epsilon(x)\right)\\
	&=
\lim_{\epsilon \downarrow 0} 
	\limsup_{k \rightarrow \infty}\frac{1}{k}\log
	P\left(\frac 1k \log(G(W_k)) \in B_\epsilon(x)\right)\\
	&=-\Lambda^*(x)
\end{align*}
and the LDP is proved.
\end{IEEEproof}

In establishing the LDP, we have shown that any rate function that
governs such an LDP must have the form of a straight line in $[0,
\TURN]$ followed by a strictly convex function. The initial straight
line comes from all words that are, in an asymptotic sense, of
greatest likelihood.

While the LDP is for the sequence $\{k^{-1}\log G(W_k)\}$, it can
be used to develop the more valuable direct estimate of the distribution
of each $G(W_k)$ found in equation \eqref{eq:wknapprox}. The next
corollary provides a rigorous statement, but an intuitive, non-rigorous
argument 
for understanding the result therein is that from the LDP we have
the approximation that for large $k$
\begin{align*}
dP\left(\frac 1k \log G(W_k) = x\right) \approx \exp(-k\Lambda^*(x)).
\end{align*}
As for large $k$ the distribution of $k^{-1} \log
G(W_k)$ and $G(W_k)/k$ are ever closer to having densities,
using the change of variables formula gives
\begin{align*}
dP\left(\frac 1k G(W_k) = x\right) 
&= \frac{1}{kx} dP\left(\frac 1k \log G(W_k) = x\right)\\
&\approx 
\frac{1}{kx}\exp\left(-k\Lambda^*\left(\frac 1k \log(kx)\right)\right).
\end{align*}
Finally, the substitution $kx=n$ gives the approximation in equation
\eqref{eq:wknapprox}. To make this heuristic precise requires
distinct means, explained in the following corollary.

\begin{corollary}[Direct estimates on guesswork]
\label{cor:curious}
Recall the definition
\begin{align*}
K_k(x,\epsilon) := 
	\left\{w\in\A^k:k^{-1}\log G(w) \in B_\epsilon(x)\right\}.
\end{align*}
For any $x\in[0,\log(m)]$ we have
\begin{align*}
&\lim_{\epsilon \downarrow 0}
\liminf_{k \rightarrow \infty}\frac 1k \log \inf_{w\in K_k(x,\epsilon)}
\P(W_k =w)\\
&=
\lim_{\epsilon \downarrow 0}
\limsup_{k \rightarrow \infty}\frac 1k \log \sup_{w\in K_k(x,\epsilon)}
\P(W_k =w)\\
&=-\left(x+\Lambda^*(x)\right).
\end{align*} 
\end{corollary}
\begin{proof}
We show how to prove the upper bound as the lower bound follows
using analogous arguments, as do the edge cases. Let $x\in(0,\log(m))$
and $\epsilon>0$ be given. Using the monotonicity of guesswork
\begin{align*}
&\limsup_{k \to \infty}\frac 1k \log \sup_{w\in K_k(x,\epsilon)}
\P(W_k =w)\\
& \leq 
\liminf_{k \to \infty}\frac 1k \log \inf_{w\in K_k(x-2\epsilon,\epsilon)}
\P(W_k =w).
\end{align*}
Using the estimate found in Theorem \ref{thm:main} and the LDP 
provides an upper bound on the latter: 
\begin{align*}
&(x-3\epsilon)+
\liminf_{k \to \infty}\frac 1k \log \inf_{w\in K_k(x-2\epsilon,\epsilon)}
\P(W_k =w)\\
&\leq 
\liminf_{k \to \infty}\frac{1}{k}\log
	P\left(\frac 1k \log(G(W_k)) \in B_\epsilon(x-2\epsilon)\right)\\
&\leq
\limsup_{k \to \infty}\frac{1}{k}\log
	P\left(\frac 1k \log(G(W_k)) \in [x-3\epsilon,x-\epsilon]\right)\\
&\leq
-\inf_{x\in [x-3\epsilon,x-\epsilon]} \Lambda^*(x).
\end{align*}
Thus
\begin{align*}
&\limsup_{k \to \infty}\frac 1k \log \sup_{w\in K_k(x,\epsilon)}
\P(W_k =w)\\
&\leq 
-x+3\epsilon -\inf_{x\in [x-3\epsilon,x-\epsilon]} \Lambda^*(x).
\end{align*}
Thus the upper-bound follows taking $\epsilon\downarrow0$ and
using the continuity when finite of $\Lambda^*$. 
\end{proof}

Unpeeling limits, this corollary shows that when $k$ is large the
probability of the $n^{\rm th}$ most likely word is approximately
$1/n \exp(-k \Lambda^*(k^{-1}\log n))$, without the need to identify
the word itself. This justifies the approximation in equation
\eqref{eq:wknapprox}, whose complexity of evaluation does not depend
on $k$. We demonstrate its merit by example in Section \ref{sec:examples}.

Before that, as a corollary to the LDP we find the
following r\^ole for the specific Shannon entropy. Thus, although
Massey established that for a given word length the Shannon entropy
is only a lower bound on the guesswork, for growing password length
the specific Shannon entropy determines the linear growth rate of
the expectation of the logarithm of guesswork.

\begin{corollary}[Shannon entropy and guesswork]
\label{cor:shannon}
Under assumptions \ref{ass:1} and \ref{ass:2},
\begin{align*}
\lim_{k \rightarrow \infty}\frac 1k \E(\log G(W_k))
	=\lim_{k\to\infty}\frac 1k H(W_k),
\end{align*}
the specific Shannon entropy.
\end{corollary}
\begin{IEEEproof}
Note that $\Lambda^*(x)=0$ if and only if $x = \Lambda'(0) =\lim
k^{-1} H(W_k)$, by arguments found in the proof of Theorem \ref{thm:main}.
The weak law then follows by concentration of measure, e.g.
\cite{Lewis95}.
\end{IEEEproof}

\section{Examples}
\label{sec:examples}

\emph{I.i.d letters}.

\begin{figure}
\includegraphics[scale=0.8]{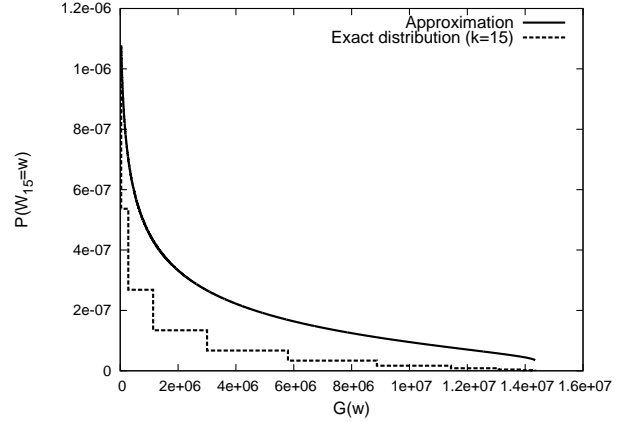}
\caption{Illustration of Corollary \ref{cor:curious}.
Words constructed from i.i.d letters with $P(W_1=1)=0.4, P(W_1=2)=0.4,
P(W_1=3)=0.2$. For $k=15$ comparison of the probability of $n^{\rm
th}$ most likely word and the approximation $1/n \exp(-k
\Lambda^*(k^{-1}\log n))$ versus $n\in\{1,\ldots,3^{15}\}$.}
\label{fig:fest}
\end{figure}

\begin{figure}
\includegraphics[scale=0.8]{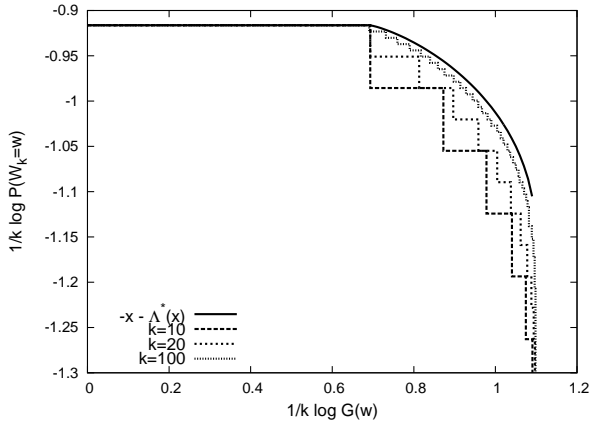}
\caption{Illustration of Corollary \ref{cor:curious}.
Words constructed from i.i.d letters with $P(W_1=1)=0.4, P(W_1=2)=0.4,
P(W_1=3)=0.2$. For $k=10,20$ and $100$, comparison of $k^{-1}$ times
the logarithm of the probability of $n^{\rm th}$ most likely word
versus $k^{-1}$ times the logarithm of $n$, as well as the approximation
$-x-\Lambda^*(x)$ versus $x$.}
\label{fig:fconv}
\end{figure}

Assume words are constructed of i.i.d. letters. Let $W_1$ take
values in $\A=\{1,\ldots,m\}$ and assume $\P(W_1=i)\ge\P(W_1=j)$
if $i\le j$. Then from \cite{Arikan96,Pfister04} and Lemma
\ref{lem:flatsCGF} we have that
\begin{align*}
\Lambda(\alpha) = 
	\begin{cases}
\displaystyle (1+\alpha) \log\sum_{w\in\A} P(W_1=w)^{1/(1+\alpha)}
	& \text{if } \alpha>-1 \\
	\log \P(W_1=1) & \text{if } \alpha\le-1.
	\end{cases}
\end{align*}
From Lemma \ref{lem:iid} we have that
\begin{align*}
\TURN=\lim_{\alpha\downarrow-1}\Lambda'(\alpha)
	\in\{0,\log(2),\ldots,\log(m)\}
\end{align*}
and no other values are possible. Unless the distribution of $W_1$
is uniform, $\Lambda^*(x)$ does not have a closed form for all $x$,
but is readily calculated numerically. With $|\A|=3$ and $k=15$,
Figure \ref{fig:fest} compares the exact distribution $P(W_k=w)$
versus $G(w)$ with the approximation found in equation \eqref{eq:wknapprox}.
As there are $3^{15}\approx 1.4$ million words, the likelihood of
any one word is tiny, but the quality of the approximation can
clearly be seen. Rescaling the guesswork and probabilities to make
them comparable for distinct $k$, Figure \ref{fig:fconv} illustrates
the quality of the approximation as $k$ grows. By $k=100$ there are
$3^{100}\approx 5.1$ times $10^{47}$ words and the underlying
combinatorial complexities of the explicit calculation become
immense, yet the complexity of calculating the approximation has
not increased.

\emph{Markovian letters}.

\begin{figure}
\includegraphics[scale=0.8]{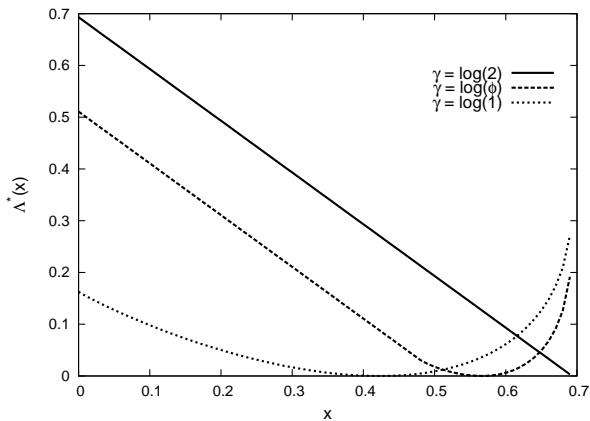}
\caption{Illustration of rate functions in Theorem \ref{thm:main}.
Words constructed from Markov letters on $|\A|=2$. Three rate
functions illustrating only values of $\TURN$ possible, $\log(1)$,
$\log(\phi)\approx 0.48$ and $\log(2)$, from Lemma \ref{lem:Markphi}.}
\label{fig:markrf}
\end{figure}

As an example of words constructed of correlated letters, consider
$\{W_k\}$ where the letters are chosen via a process a Markov chain
with transition matrix $P$ and some initial distribution on $|\A|=2$.
Define the matrix $P_\alpha$ by $(P_\alpha)_{i,j}=p_{i,j}^{1/(1+\alpha)}$,
then by \cite{Malone04,Pfister04} and Lemma \ref{lem:flatsCGF} we
have that
\begin{align*}
\Lambda(\alpha) = 
	\begin{cases}
	(1+\alpha)\log \rho(P_\alpha) & \text{if } \alpha>-1 \\
	\log\max(p_{1,1},p_{2,2},\sqrt{p_{1,2}\, p_{2,1}}) 
		& \text{if } \alpha\le-1,
	\end{cases}
\end{align*}
where $\rho$ is the spectral radius operator. In the two letter
alphabet case, with $\beta=1/(1+\alpha)$ we have that
$\rho(P_{(1-\beta)/\beta})$ equals
\begin{align*}
\frac{p_{1,1}^{\beta}+p_{2,2}^{\beta}}{2} 
+ \frac{\sqrt{(p_{1,1}^{\beta} - p_{2,2}^{\beta})^2
        + 4(1-p_{2,2})^{\beta}(1-p_{1,1})^{\beta}}}{2}.
\end{align*}
As with the i.i.d. letters example, apart from in special cases,
the rate function $\Lambda^*$ cannot be calculated in closed form,
but is readily evaluated numerically. Regardless, we have the following,
perhaps surprising, result on the exponential rate of growth of the
size of the set of almost most likely words.

\begin{lemma}[The Golden Ratio and Markovian letters]
\label{lem:Markphi}
For $\{W_k\}$ constructed of Markovian letters,
\begin{align*}
\TURN=\lim_{\alpha\downarrow-1}\Lambda'(\alpha)
	\in\{0,\log(\phi),\log(2)\},
\end{align*}
where $\phi= (1+\sqrt{5})/2$ is the Golden Ratio,
and no other values
are possible.
\end{lemma}

This lemma can be proved by directly evaluating the derivative of
$\Lambda(\alpha)$ with respect to $\alpha$. Note that here $\exp(k
\TURN)$ definitely only describes the number of words of equal
highest likelihood when $k$ is large as the initial distribution
of the Markov chain plays no r\^ole in $\TURN$'s evaluation.

The case where $\TURN=\log(2)$ occurs when $p_{1,1}=p_{2,2}=1/2$.
The most interesting case is when there are approximately $\phi^k$
approximately equally most likely words. This occurs if
$p_{1,1}=\sqrt{p_{1,2}p_{2,1}}>p_{2,2}$. For large $k$, words of
near-maximal probability have the form of a sequence of 1s, where
a 2 can be inserted anywhere so long as there is a 1 between it and
any other 2s. A further sub-exponential number of aberrations are
allowed in any given sequence. For example, with an equiprobable
initial distribution and $k=4$ there are $8$ most likely words
(1111, 1112, 1121, 1211, 1212, 2111, 2121, 2112) and $\phi^4\approx
6.86$.

Figure \ref{fig:markrf} gives plots of $\Lambda^*(x)$ versus $x$
illustrating the full range of possible shapes that rate functions
can take: linear, linear then strictly convex, or strictly convex,
based on the transition matrices
\begin{align*}
\left(
	\begin{matrix}
	0.5 & 0.5 \\
	0.5 & 0.5 
	\end{matrix}
\right),
\left(
	\begin{matrix}
	0.6 & 0.4 \\
	0.9 & 0.1 
	\end{matrix}
\right)
\text{ and }
\left(
	\begin{matrix}
	0.85 & 0.15 \\
	0.15 & 0.85 
	\end{matrix}
\right)
\end{align*}
respectively.


\begin{thebibliography}{10}
\providecommand{\url}[1]{#1}
\csname url@samestyle\endcsname
\providecommand{\newblock}{\relax}
\providecommand{\bibinfo}[2]{#2}
\providecommand{\BIBentrySTDinterwordspacing}{\spaceskip=0pt\relax}
\providecommand{\BIBentryALTinterwordstretchfactor}{4}
\providecommand{\BIBentryALTinterwordspacing}{\spaceskip=\fontdimen2\font plus
\BIBentryALTinterwordstretchfactor\fontdimen3\font minus
  \fontdimen4\font\relax}
\providecommand{\BIBforeignlanguage}[2]{{%
\expandafter\ifx\csname l@#1\endcsname\relax
\typeout{** WARNING: IEEEtran.bst: No hyphenation pattern has been}%
\typeout{** loaded for the language `#1'. Using the pattern for}%
\typeout{** the default language instead.}%
\else
\language=\csname l@#1\endcsname
\fi
#2}}
\providecommand{\BIBdecl}{\relax}
\BIBdecl

\bibitem{Massey94}
J.~L. Massey, ``Guessing and entropy,'' \emph{Proc. IEEE Int. Symp. Inf.
  Theory}, pp. 204--204, 1994.

\bibitem{Arikan96}
E.~Arikan, ``An inequality on guessing and its application to sequential
  decoding,'' \emph{IEEE Trans, Inf. Theory}, vol.~42, pp. 525--526, 1996.

\bibitem{Malone04}
D.~Malone and W.~G. Sullivan, ``Guesswork and entropy,'' \emph{IEEE Trans. Inf.
  Theory}, vol.~50, no.~4, pp. 525--526, 2004.

\bibitem{Pfister04}
C.-E. Pfister and W.~G. Sullivan, ``R\'{e}nyi entropy, guesswork moments and
  large deviations,'' \emph{IEEE Trans. Inf. Theory}, vol.~50, no.~11, pp.
  2794--2800, 2004.

\bibitem{Hanawal11}
M.~K. Hanawal and R.~Sundaresan, ``Guessing revisited: A large deviations
  approach,'' \emph{IEEE Trans. Inf. Theory}, vol.~57, no.~1, pp. 70--78, 2011.

\bibitem{Dembo}
A.~Dembo and O.~Zeitouni, \emph{Large Deviations Techniques and
  Applications}.\hskip 1em plus 0.5em minus 0.4em\relax Springer, 2009.

\bibitem{Rockafellar70}
R.~T. Rockafellar, \emph{Convex analysis}, ser. Princeton Mathematical Series,
  No. 28.\hskip 1em plus 0.5em minus 0.4em\relax Princeton, N.J.: Princeton
  University Press, 1970.

\bibitem{Baldi88}
P.~Baldi, ``Large deviations and stochastic homogenization,'' \emph{Ann. Mat.
  Pura Appl. (4)}, vol. 151, pp. 161--177, 1988.

\bibitem{Lewis95A}
J.~T. Lewis and C.-E. Pfister, ``Thermodynamic probability theory: some aspects
  of large deviations,'' \emph{Russian Math. Surveys}, vol.~50, no.~2, pp.
  279--317, 1995.

\bibitem{Lewis95}
J.~T. Lewis, C.-E. Pfister, and W.~G. Sullivan, ``Entropy, concentration of
  probability and conditional limit theorems,'' \emph{Markov Process. Related
  Fields}, vol.~1, no.~3, pp. 319--386, 1995.

\end{thebibliography}
\end{document}